\documentclass[letterpaper, 10pt, conference]{ieeeconf}  %

\IEEEoverridecommandlockouts                              %

\usepackage{amsmath} %
\usepackage{amssymb}  %
\usepackage[noadjust]{cite}
\usepackage{graphicx}
\usepackage{pgfplots}
\pgfplotsset{compat=newest}
\usepgfplotslibrary{external}
\usepackage{bbm}
\usepackage{mathtools}
\usepackage{dsfont}
\usepackage{bm}					%

\makeatletter
\let\NAT@parse\undefined
\makeatother
\usepackage[hidelinks]{hyperref} %

\newtheorem{thm}{Theorem}

\newtheorem{lem}{Lemma}

\newtheorem{conj}{Conjecture}

\newtheorem{prop}{Proposition}

\newtheorem{remark}{Remark}

\newcommand*\TTT{\ensuremath{\tau_{\mathrm{per}}}}
\newcommand*\Tinf{\ensuremath{\tau_\infty}}
\newcommand*\Ttwo{\ensuremath{\tau_2}}
\newcommand{\E}[1]{\ensuremath{\mathbb{E}\!\left\lbrack#1\right\rbrack}}
\newcommand{\V}[1]{\ensuremath{\mathbb{V}\!\left\lbrack#1\right\rbrack}}
\newcommand*\Prob{\ensuremath{\mathbb{P}}}
\newcommand*\diff{\mathop{}\!\mathrm{d}}

\definecolor{ceruleanblue}{rgb}{0.16, 0.32, 0.75}
\definecolor{tangelo}{rgb}{0.98, 0.3, 0.0}
\definecolor{aoenglish}{rgb}{0.0, 0.5, 0.0}
\definecolor{chromeyellow}{rgb}{1.0, 0.65, 0.0}
\definecolor{darkorchid}{rgb}{0.6, 0.2, 0.8}

\title{\LARGE \bf
	Performance implications of different p-norms in level-triggered sampling*
\thanks{*F.\ Allgöwer thanks the German Research Foundation (DFG) for support of this work within grant AL 316/13-2 and within the German Excellence Strategy under grant EXC-2075 - 285825138; 390740016.
	}}

\author{David Meister$^{1}$, 
and Frank Allgöwer$^{1}$%
\thanks{$^{1}$D.\ Meister and F.\ Allgöwer are with the University of Stuttgart, Institute for Systems Theory and Automatic Control, Stuttgart, Germany,\newline
		{\tt\small \{meister, allgower\}@ist.uni-stuttgart.de}}%
}

\begin{document}

\pubid{\begin{minipage}{\textwidth}\ \\[12pt] \copyright 2023 IEEE. Personal use of this material is permitted. Permission from IEEE must be obtained for all other uses, in any current or future media, including reprinting/republishing this material for advertising or promotional purposes, creating new collective works, for resale or redistribution to servers or lists, or reuse of any copyrighted component of this work in other works.\end{minipage}}

\maketitle

\begin{abstract}
	This work studies the performance of an event-based control approach, namely level-triggered sampling, in a standard multidimensional single-integrator setup.
	We falsify a conjecture from the literature that the deployed $\bm{p}$-norm in the triggering condition supposedly has no impact on the performance of the sampling scheme in that setting.
	In particular, we show for the considered setup that the usage of the maximum norm instead of the Euclidean norm induces a performance deterioration of level-triggered sampling for sufficiently large system dimensions, when compared to periodic control at the same sampling rate.
	Moreover, we investigate the performance for other $\bm{p}$-norms in simulation and observe that it degrades with increasing $\bm{p}$.
	In addition, our findings reveal the previously unknown role of the triggering rule in the cause of a recently discovered phenomenon:
	Previous work has shown for a single-integrator consensus setup that the commonly observed performance advantage of event-based control over periodic control can be lost in distributed settings with a cooperative control goal.
	In our work, we obtain similar results for a non-cooperative setting only by adjusting the norm in the level-triggered sampling scheme.
	We therefore demonstrate that the performance degradation found in the distributed setting originates from the triggering rule and not from the considered cooperative control goal.
\end{abstract}

\begin{keywords}
	Networked control systems, Event-triggered control, Sampled-data control, Control over communications.
\end{keywords}

\section{Introduction}

Instead of sampling periodically, event-based control closes the feedback loop only if an event-triggering condition is satisfied.
While periodic control still is the typical choice for sample-based control systems, event-based control schemes are a powerful alternative with relevance in various applications such as networked control systems or systems with energy efficiency requirements.
The advantage of event-based control for applications requiring resource-aware control loop design \cite{Lunze2018} lies in its capability to reduce the usage rate of communication networks or energy-consuming components in the control loop while maintaining a similar performance level when compared to periodic control.

Early results by \cite{Astrom2002} have shown this potential in a single-integrator setup with level-triggered sampling.
Subsequent works such as \cite{Meng2012} have generalized the findings to 2-dimensional integrator systems.
As it remained challenging to provide theoretical guarantees on the performance advantage of event-based control for more general settings, \cite{Antunes2016,Antunes2018,Khashooei2018,Antunes2020,Balaghiinaloo2021,Balaghiinaloo2022} have introduced and examined the concept of consistency for event-based control approaches.
In short, an event-based control scheme is called consistent if it provides an equal or better level of performance as any periodic controller while operating at the same average sampling rate.
The authors propose design techniques to arrive at consistent event-based control schemes for LQ- or $\mathcal{L}_2/\ell_2$-performance measures and linear time-invariant systems.

Another line of research on event-based control is aimed at finding optimal triggering rules and control laws with respect to predefined performance measures which formulate a trade-off between control performance and triggering rate.
Typical performance measures consider quadratic state costs and linear triggering rate costs including a scalar trade-off factor, such as in \cite{Molin2009} for a finite horizon or in \cite[Paper~II]{Henningsson2012} for an infinite horizon problem.
The authors of the latter paper provide a closed-form solution for an optimal triggering rule in the multidimensional integrator case.
Furthermore, the work \cite{Andren2017} proposes a numerical sampling scheme design method for an LQG setting with output feedback.
It builds upon the results in \cite{Mirkin2017,Braksmayer2017,Goldenshluger2017} which derive an $\mathcal{H}_2$-optimal controller design approach for arbitrary uniformly bounded sampling patterns and show that the design of optimal controller and optimal triggering rule are separable in this setting.
Moreover, the authors in \cite{Mi2022} present an event-based controller design approach for continuous-time LTI systems with $\mathcal{H}_\infty$-performance guarantees.
They also prove consistency of the proposed event-based controller with respect to the optimal periodic sample-based controller.

\pubidadjcol

Inspired by \cite[Rem.~1]{Goldenshluger2017}, we study the performance of level-triggered sampling schemes for an $n$-dimensional single-integrator system as considered in \cite[Paper~II]{Henningsson2012} and \cite{Goldenshluger2017}, but with different $p$-norms in the triggering condition.
We thereby falsify the conjecture in \cite[Rem.~1]{Goldenshluger2017} that the chosen $p$-norm supposedly plays no role for the performance of the level-triggered sampling scheme in this setting.
While the $2$-norm level-triggering scheme is well studied and outperforms periodic control for the same average sampling rate, see e.g., \cite{Henningsson2012,Goldenshluger2017,Mirkin2017}, our theoretical analysis shows a performance degradation for the $\infty$-norm case.
For sufficiently large system dimensions and equal average sampling rates, we find that periodic control outperforms event-based control with $\infty$-norm triggering.
Moreover, we examine the performance implications of other $p$-norms in the level-triggering condition in simulation.
We observe that the performance of the level-triggered sampling scheme degrades with increasing $p$.
In addition, we obtain new insights into a recently discovered phenomenon:
Our preliminary work \cite{Meister2022,Meister2023} shows for a distributed single-integrator consensus problem that the performance advantage of event-based control compared to periodic control might not generally be provided in cooperative setups.
With the results in this paper, we demonstrate that the decentralized level-triggering rule in \cite{Meister2022}, which can also be written in terms of the $\infty$-norm triggering rule, sits at the core of the found phenomenon.
This is due to the fact that, in this work, we consider the same triggering rule in a non-cooperative instead of a cooperative setting and still find the described performance degradation of event-based control.
We thereby demonstrate that the phenomenon is not a direct consequence of considering a cooperative control goal, but that further research on it should focus on examining decentralized triggering rules and their performance implications.

Our paper is structured as follows:
In Section~\ref{sec:setup}, we describe the setup and introduce the considered problem.
After that, we present our theoretical results in Section~\ref{sec:analysis}, and support our findings with numerical simulations in Section~\ref{sec:sim}.
We conclude this work in Section~\ref{sec:conclusion}.

\section{Setup and Problem Formulation}\label{sec:setup}

We consider a system that can be described by the following dynamical equation
\begin{equation}\label{eq:system}
	\diff x(t) = u(t) \diff t + \diff v(t),
\end{equation}
where $x(t)=\lbrack x_1(t),\dots,x_n(t)\rbrack^\top$ refers to the system state, $u(t)=\lbrack u_1(t),\dots,u_n(t)\rbrack^\top$ to the control input, and $v(t)=\lbrack v_1(t),\dots,v_n(t)\rbrack^\top$ to an $n$-dimensional standard Wiener process.

Our aim is to quantify and relate the performance of a selection of sample-based control algorithms, namely level-triggered and periodic control schemes.
In order to facilitate a fair comparison, we will do so under equal average sampling rates.
Note that this comparison concept is equivalent to a comparison of average sampling rates under equal performance requirements \cite{Astrom2002}.
This perspective has for example also been leveraged in \cite{Astrom2002,Henningsson2012,Antunes2016,Antunes2018,Antunes2020,Balaghiinaloo2021,Balaghiinaloo2022}.

The control goal is to keep the described system as close to the origin as possible over time.
Therefore, we define the following cost functional as performance measure
\begin{equation}\label{eq:cost}
	J := \limsup_{M\to \infty} \frac{1}{M} \int_{0}^{M} \E{x(t)^\top x(t)} \diff t,
\end{equation}
which quantifies the long-term average of the expected quadratic deviation of the system state from the origin, or, in other words, the asymptotic average variance of $x(t)$.

Consequently, the input
\begin{equation}\label{eq:input}
	u(t) = -\sum_{k\in\mathbb{N}} \delta(t-t_{k}) x(t_{k})
\end{equation}
is optimal under the considered performance measure \eqref{eq:cost} regardless of the deployed sampling scheme.
We denote by $\delta(\cdot)$ the Dirac impulse and by $(t_k)_{k\in\mathbb{N}_0}$ with $t_0=0$ the sampling instants determined by the respective sampling scheme.
The defined impulsive control input resets the system to the origin at every sampling instant.

Note that this setup has been studied in various flavors in the literature on event-based control, see e.g., \cite{Astrom2002,Rabi2009,Meng2012,Henningsson2012,Goldenshluger2017}. %
We contribute to this collection on fundamental characteristics of event-based control schemes in this standard setup with the following two objectives:
Firstly, we examine the impact of the selected $p$-norm in a well-studied event-triggering scheme - level-triggered sampling - on the performance of the closed loop and, thereby, investigate the following conjecture from \cite[Rem.~1]{Goldenshluger2017}.
\begin{conj}[\cite{Goldenshluger2017}] \label{conj:p_norm}
	Monte Carlo simulations indicate that the performance of level-triggered sampling schemes is independent of the utilized $p$-norm in the triggering rule.
\end{conj}

Secondly, the studied $\infty$-norm level-triggering rule coincides with the decentralized level-triggering rule in \cite{Meister2022} whereas we study a non-cooperative instead of a cooperative control goal here.
Due to this connection, we attain new insights on the phenomenon found in \cite{Meister2022} that periodic control outperforms event-based control in a particular distributed consensus setup if the number of agents is large enough.

\section{Performance Analysis}\label{sec:analysis}

In this section, we present our performance analysis results for the previously described setup.
After some preliminaries, we present a performance result for periodic control as a comparison baseline.
Subsequently, we state our performance results for the event-based control schemes considered in Conjecture~\ref{conj:p_norm}. %
Since general performance results cannot be obtained as of now, we consider the two special cases of Euclidean and maximum norm triggering.

\subsection{Preliminaries}

Let us state helpful facts about the considered problem.
\begin{lem}\label{lem:0T}
	Let the inter-event times be independent and identically distributed and let $\E{\tau}<\infty$, where $\tau$ is the inter-event time determined by the sampling scheme.
	Then, the cost \eqref{eq:cost} can be computed according to
	\begin{equation*}
		J(\tau) = \frac{\E{\int_{0}^{\tau} x(t)^\top x(t) \diff t}}{\E\tau},
	\end{equation*}
	where we only need to consider the first sampling interval.
\end{lem}
\begin{proof}
	The proof follows along the lines of \cite[Lem.~1]{Meister2023} and is omitted due to space limitations.
\end{proof}

Let us define $Q(\tau) := \E{\int_{0}^{\tau} x(t)^\top x(t) \diff t}$.
\begin{lem}\label{lem:Q}
	Let $\tau$ be independent of the direction, i.e., $\tau$ does not change if $v_i$ is interchanged with $v_j$ for any $i,j\in\lbrace1,\dots,n\rbrace$.
	Then, given Lemma~\ref{lem:0T}, we can establish
	\begin{equation*}
		Q(\tau) = n \, \E{\int_{0}^{\tau} v_1(t)^2 \diff t}.
	\end{equation*}
\end{lem}
\vspace*{\belowdisplayskip}
\begin{proof}
	We have
	\begin{align*}
		Q(\tau) & = \E{\int_0^{\tau} \sum_{i=1}^n x_i^2 \diff t} = \E{\int_0^{\tau} \sum_{i=1}^n v_i^2 \diff t} \\
		        & = n \, \E{\int_0^{\tau} v_1(t)^2  \diff t},
	\end{align*}
	where we used that $\tau$ is independent of the direction in the last step.
\end{proof}

\begin{remark}
	The assumption formulated in Lemma~\ref{lem:Q} is satisfied for the considered sampling schemes in this paper.
\end{remark}

\subsection{Periodic Control}\label{sec:TT}

For periodic control, the inter-event time is a user-defined constant $\TTT = t_{k+1} - t_k = \mathrm{const}$.
As this sequence of inter-event times is indeed independent and identically distributed as well as independent of the direction in the sense of Lemma~\ref{lem:Q}, we are able to state the following result.
\begin{prop}\label{prop:TT_cost}
	Suppose system \eqref{eq:system} is controlled by the control input \eqref{eq:input} with constant inter-event times $\TTT$.
	Then, the cost \eqref{eq:cost} can be expressed as
	\begin{equation*}
		J_{\mathrm{per}}(\TTT) = n \frac{\TTT}{2}.
	\end{equation*}
\end{prop}
\vspace*{\belowdisplayskip}
\begin{proof}
	As the inter-event times $\TTT$ are constant, Lemmas~\ref{lem:0T} and \ref{lem:Q} apply.
	Thus, we arrive at
	\begin{align*}
		Q(\TTT) & = n \int_{0}^{\TTT} \E{v_1(t)^2} \diff t            \\
		        & = n \int_{0}^{\TTT} t \diff t = n \frac{\TTT^2}{2},
	\end{align*}
	which we can use in $J_{\mathrm{per}}(\TTT) = Q(\TTT)/\TTT$ to arrive at the desired expression for \eqref{eq:cost}.
\end{proof}

\begin{remark}
	Note that this result is in line with the ones found in \cite{Astrom2002,Rabi2009,Blind2011}, scaled by the state dimension.
\end{remark}

The found cost expression will serve as a baseline for the event-based control schemes analyzed in the following section.

\subsection{Event-Based Control}\label{sec:ET}

For this paper, let us consider $p$-norm level-triggering rules of the form
\begin{equation}\label{eq:ET_cond}
	\lVert x(t) \rVert{}_p \geq \Delta_p
\end{equation}
where $\Delta_p>0$ is a constant and $p\in[1,\infty]$.
This is motivated by Conjecture~\ref{conj:p_norm}.
In addition, note that the decentralized triggering rule considered in \cite{Meister2022} for a consensus problem is equal to \eqref{eq:ET_cond} with $p=\infty$.

The inter-event times resulting from this type of triggering rule can be defined as a stopping time
\begin{equation*}
	\tau_p(\Delta_p) = \inf\{ t>0 \;:\; \lVert x(t) \rVert{}_p \geq \Delta_p \}.
\end{equation*}
Note that all inter-event times are independent and identically distributed.
Given the considered type of triggering rule, we can establish the following lemma.
\begin{lem}\label{lem:scaling}
	Given the triggering condition \eqref{eq:ET_cond}, the following scaling relationships hold
	\begin{align*}
		Q_p(\Delta_p)        & = \Delta_p^4 Q_p(1),        \\
		\E{\tau_p(\Delta_p)} & = \Delta_p^2 \E{\tau_p(1)}, %
	\end{align*}
	where $Q_p(\Delta_p) := Q(\tau_p(\Delta_p))$.
\end{lem}
\begin{proof}
	See Appendix~\ref{app:proof_scaling}.
\end{proof}

The scaling relationships allow us to concentrate on the case $\Delta_p=1$ for the performance analysis of the triggering rules \eqref{eq:ET_cond}.
This is due to the following lemma which utilizes the shorthand $J_p(\Delta_p):=J(\tau_p(\Delta_p))$.
\begin{lem}\label{lem:DeltaEq1}
	Given the assumptions of Proposition~\ref{prop:TT_cost} and triggering condition \eqref{eq:ET_cond}, the following holds
	\begin{equation*}
		\frac{J_p(\Delta_p)}{J_\mathrm{per}(\E{\tau_p(\Delta_p)})} = \frac{J_p(1)}{J_\mathrm{per}(\E{\tau_p(1)})}.
	\end{equation*}
\end{lem}
\vspace*{\belowdisplayskip}
\begin{proof}
	Utilizing Lemma~\ref{lem:scaling} with Proposition~\ref{prop:TT_cost} and Lemma~\ref{lem:0T} yields
	\begin{align*}
		J_p(\Delta_p)                        & = \Delta_p^2 J_p(1),                        \\
		J_\mathrm{per}(\E{\tau_p(\Delta_p)}) & = \Delta_p^2 J_\mathrm{per}(\E{\tau_p(1)}),
	\end{align*}
	which allows us to cancel $\Delta_p$ from the considered ratio.
\end{proof}

Thus, the properties shown for $J_p(1) / J_\mathrm{per}(\E{\tau_p(1)})$ also hold for the respective ratio with any $\Delta_p>0$.
For the remainder of this work, we will thus omit the argument $\Delta_p$ whenever a result holds for any choice of $\Delta_p>0$.

Let us now proceed with our analysis for two specific choices of $p$, namely $p=2$ and $p=\infty$.
The $2$-norm case has already been studied in the literature, and we present reformulated results such that a comparison with periodic control and other $p$-norm level-triggering schemes becomes possible.
For the $\infty$-norm case, we deduce new theoretical findings and put them into context with the presented 2-norm results and Conjecture~\ref{conj:p_norm}.

\subsubsection{Euclidean Norm Triggering}

Utilizing the Euclidean or $2$-norm for triggering results in the inter-event times
\begin{equation}\label{eq:ET_IET_2}
	\Ttwo(\Delta_2) = \inf\lbrace t>0 \;:\; \lVert x(t) \rVert{}_2 \geq \Delta_2 \rbrace.
\end{equation}
As laid out in the introduction, this case is well-studied in the literature, see e.g., \cite{Astrom2002,Meng2012,Henningsson2012,Goldenshluger2017}.
Let us now leverage existing results to deduce performance characteristics that we can use in our intended comparison.

Following the arguments in \cite[Paper~II]{Henningsson2012} or \cite{Goldenshluger2017}, we arrive at a closed-form solution for the cost ratio of 2-norm level-triggering and periodic control.
\begin{prop}\label{prop:ET_cost_2}
	Given the triggering rule in \eqref{eq:ET_IET_2}, we have
	\begin{equation*}
		\frac{J_2}{J_\mathrm{per}(\E{\Ttwo})} = \frac{n}{n+2},
	\end{equation*}
	where $J_\mathrm{per}(\E\Ttwo)$ denotes the periodic control cost with $\TTT = \E\Ttwo$.
\end{prop}
\begin{proof}
	The result follows, for example, from \cite[Prop.~2]{Goldenshluger2017} with $\Sigma_v,\Sigma_u\rightarrow0$.
	We omit the details here.
\end{proof}

\begin{remark}
	We additionally know from \cite[Paper~II]{Henningsson2012} that, for a fixed average sampling rate, the 2-norm level-triggering rule is optimal with respect to \eqref{eq:cost}.
	Given this result, we can reformulate Conjecture~\ref{conj:p_norm} as the hypothesis that any $p$-norm utilized in the considered triggering condition is supposedly optimal in the sense that the resulting triggering rule minimizes \eqref{eq:cost} for a fixed $\E{\tau_p}$.
\end{remark}

\subsubsection{Maximum Norm Triggering}

Let us now examine Conjecture~\ref{conj:p_norm} for the $\infty$- or maximum norm.
The corresponding inter-event time can be formulated as
\begin{equation}\label{eq:ET_IET_inf}
	\Tinf(\Delta_\infty) = \inf\lbrace t>0 \;:\; \lVert x(t) \rVert{}_\infty \geq \Delta_\infty \rbrace.
\end{equation}
Before we can arrive at a result on the performance relationship between event-based and periodic control within this setting, we require the following lemma.
\begin{lem}\label{lem:ET_inf_var}
	Let $Z, (Z_n)$ be a random variable and a sequence of random variables, respectively.
	Moreover, let $(Z_n)$ converge weakly to $Z$ and $\V Z < \infty$, where $\V\cdot$ denotes the variance.
	Then,
	\begin{equation*}
		\liminf_{n\to\infty} \V{Z_n} \geq \V Z.
	\end{equation*}
\end{lem}
\vspace*{\belowdisplayskip}
\begin{proof}
	See Appendix~\ref{app:proof_ET_inf_var}
\end{proof}

Given this lemma, we can now show the following result.

\begin{thm}\label{thm:ET_cost_inf}
	Given that system \eqref{eq:system} is controlled with the impulsive control input \eqref{eq:input} under the triggering rule in \eqref{eq:ET_IET_inf}, there exists an $n_0$ such that, for all $n\geq n_0$, we have
	\begin{equation*}
		J_\mathrm{per}(\E\Tinf) < J_\infty
	\end{equation*}
	where $J_\mathrm{per}(\E\Tinf)$ denotes the periodic control cost with $\TTT = \E\Tinf$.
\end{thm}
\begin{proof}
	Due to Lemma~\ref{lem:DeltaEq1}, we only consider the case $\Delta_\infty=1$.
	In this proof, we will therefore omit the arguments of $\Tinf(1)$ and $Q_\infty(1)$ for better readability.

	As a first step, we can follow similar arguments as in the proof of \cite[Lem.~1]{Meister2022} to show that
	\begin{equation} \label{eq:limittheorem}
		2(\ln n)^2 \left({\Tinf}- a_n\right) \Rightarrow G,  \qquad \textrm{as } n\to\infty,
	\end{equation}
	with
	\begin{equation*}
		a_n := \frac{1}{2\ln n}-\frac{\ln \frac{\sqrt{2/\pi}}{(2\ln n)^{1/2}}}{2(\ln n)^2},
	\end{equation*}
	and where $\Rightarrow$ denotes convergence in distribution.
	Furthermore, $G$ is a Gumbel-distributed random variable, i.e.,
	\begin{equation*}
		\Prob( G\ge r) = \exp( - \exp(r)).
	\end{equation*}
	In addition to the direct proof in \cite[Lem.~1]{Meister2022}, \eqref{eq:limittheorem} can also be derived from \cite[Thm.~2.1.6]{Galambos1978}.
	We will therefore omit the details here.

	With Lemma~\ref{lem:ET_inf_var}, there exists an $n_1$ such that
	\begin{equation*}
		\V{2(\ln n)^2 \left({\Tinf}- a_n\right)} \geq 0.5\cdot\V G \quad \forall n\geq n_1,
	\end{equation*}
	which implies
	\begin{equation}\label{eq:var_ineq}
		\V{\Tinf} \geq \frac{1}{2}\cdot\frac{\V G}{4 (\ln n)^4} \quad \forall n\geq n_1,
	\end{equation}
	with $\V G = \pi^2/6$.
	Following a similar reasoning as in the proof of \cite[Thm.~2]{Meister2022}, we arrive at
	\begin{equation*}
		\E{\int_0^{\Tinf} v_1(t)^2 \diff t} > \frac{\E{\Tinf}^2}{2} + \frac{\V{\Tinf}}{2}  -  c \, n^{-1/2},
	\end{equation*}
	with $c=\mathbb{E}[ (\int_0^{T_2} v_1(t)^2 \diff t)^2 ]^{1/2}$ and $T_j := \inf\lbrace t>0 : \lvert x_j(t)\rvert = 1 \rbrace$ for all $j\in\lbrace1,\dots,n\rbrace$.
	Furthermore, there exists an $n_2$ such that $\frac{1}{4}\cdot \frac{\pi^2/24}{(\ln n)^{4}} -  c \, n^{-1/2} > 0$ for all $n\geq n_2$ and let $n_0\coloneqq\max(n_1,n_2)$.
	Together with \eqref{eq:var_ineq}, we have for $n\geq n_0$
	\begin{align*}
		\frac{1}{n}\, Q_\infty & =  \E{\int_0^{\Tinf} v_1(t)^2 \diff t}
		\\
		                       & >
		\frac{\E{\Tinf}^2}{2} + \frac{1}{4}\cdot \frac{\pi^2/24}{(\ln n)^{4}} -  c \, n^{-1/2}
		\\
		                       & > \frac{\E{\Tinf}^2}{2}
		=  \frac{1}{n}\, Q(\TTT=\E{\Tinf}),
	\end{align*}
	where we used Lemma~\ref{lem:Q} in the first step and Proposition~\ref{prop:TT_cost} in the last step.
	Multiplying both sides with $n/\E{\Tinf}$ gives the desired result.
\end{proof}

\begin{remark}
	Note that, in the proof of Theorem~\ref{thm:ET_cost_inf}, we have used a similar proof technique as in \cite{Meister2022,Meister2023} while we consider a different, i.e., non-cooperative, setup.
	In addition, we provide a much shorter proof for the desired inequality by leveraging Lemma~\ref{lem:ET_inf_var} instead of deriving the asymptotic order of the moments of $\Tinf$.
	The latter requires to show the convergence of the first and second moment in addition to the convergence in distribution in \eqref{eq:limittheorem}.
\end{remark}

With Theorem~\ref{thm:ET_cost_inf}, we are able to falsify Conjecture~\ref{conj:p_norm} for the $\infty$-norm.
In particular, we prove that triggering with the $\infty$-norm performs worse than periodic control for $n\geq n_0$ while triggering with the $2$-norm outperforms periodic control for any state dimension.
Thus, the Euclidean and maximum norm triggering schemes cannot result in the same performance as suspected in Conjecture~\ref{conj:p_norm}.
While it remains an open problem to examine the conjecture for $p$-norms other than $p=2$ and $p=\infty$, we can already conclude that the choice of the norm matters for the two examined cases.

In addition and as already mentioned, the decentralized triggering condition considered in \cite{Meister2022} can also be formulated via the maximum norm if all agent states are stacked in a vector.
While we consider a non-cooperative control goal instead of a consensus problem in this work, we find that periodic control outperforms $\infty$-norm triggering for sufficiently large system dimensions $n$, similarly to the phenomenon discovered in \cite{Meister2022}.
This observation as well as the underlying proofs indicate that the decentralized triggering condition is the key component inducing the found phenomenon.
Thereby, we also demonstrate that this phenomenon is not a direct consequence of considering a cooperative control goal.
Instead, we can conclude that the form of the deployed triggering condition and the information available to each agent are at the root of the observed twist in the performance relationship of periodic and event-based control.
Thus, future research can focus on whether there exists a consistent decentralized triggering rule and, if yes, what form it has.

\section{Simulation}\label{sec:sim}

After demonstrating the difference between $2$- and $\infty$-norm triggering from a theoretical perspective, we support these findings with Monte Carlo simulations in this section.
In addition, we will evaluate performance ratios for other $p$-norms and, thereby, try to attain more insights for which we have not presented a theoretical analysis yet.

We therefore perform Monte Carlo simulations in order to obtain estimates for the cost ratio $J_p/J_\mathrm{per}(\E{\tau_p})$ for various $p\in[1,\infty]$.
Note once more that, due to Lemma~\ref{lem:DeltaEq1}, the performance ratios obtained in simulation are valid for any choice of $\Delta_p>0$.
Thus, without loss of generality, we choose $\Delta_p=1$ in our simulations.
Given the simulation results for the event-based control schemes, we can then estimate the cost ratio under perfect satisfaction of the constraint $\tau_\mathrm{per}=\E{\tau_p}$ by leveraging Proposition~\ref{prop:TT_cost}.
In addition, since Lemma~\ref{lem:0T} applies, we can terminate a Monte Carlo run when the first event has been reached.
We can therefore estimate the cost ratio based on the following result.
\begin{prop}\label{prop:sim}
	Let us consider system dynamics \eqref{eq:system} with control input \eqref{eq:input} and let $\tau$ be any stopping time satisfying the assumptions of Lemmas~\ref{lem:0T} and \ref{lem:Q}.
	Then, the cost ratio $J(\tau)/J_\mathrm{per}(\E{\tau})$ can be expressed as
	\begin{equation*}
		\frac{J(\tau)}{J_\mathrm{per}(\E{\tau})} = \frac{\E{R(\tau)^4}}{n(n+2)\E{\tau}^2}
	\end{equation*}
	where $R(t):=\lVert v(t) \rVert{}_2$ is a Bessel process of dimension $n$ started at $R(0)=0$.
\end{prop}
\begin{proof}
	See Appendix~\ref{app:proof_sim}.
\end{proof}

We simulate the system utilizing the Euler-Maruyama method with a step size of $10^{-4}\,\mathrm{s}$.
Moreover, we perform $20\,000$ Monte Carlo runs per experiment to estimate the required expected values in the cost ratio expression from Proposition~\ref{prop:sim}.

\begin{figure}
	\centering
	\input{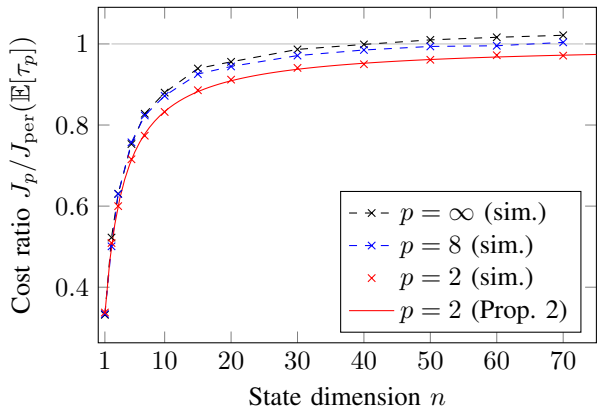}
	\setlength{\abovecaptionskip}{0pt}
	\caption{Cost ratio $J_p/J_\mathrm{per}(\E{\tau_p})$ for different $p$-norms in the triggering condition.
		Simulation results are interpolated with dashed lines.}
	\label{fig:cost_ratio}
\end{figure}

The simulation results are shown in Fig.~\ref{fig:cost_ratio}.
We have simulated the system under $2$-, $8$- and $\infty$-norm triggering conditions for various state dimensions.
In addition, we also depict the closed-form solution of the cost ratio given in Proposition~\ref{prop:ET_cost_2} for the $2$-norm case.
Firstly, we can validate the simulation results by comparing them to the closed-form solution for the $2$-norm triggering condition.
On the one hand, we observe that the estimated cost ratio is close to the closed-form solution for the state dimensions considered in the simulation.
On the other hand, there remains some deviation resulting from the finite number of Monte Carlo samples per run and the utilized step size.

Secondly, we are able to confirm our theoretical result that, with the $\infty$-norm triggering condition, periodic control outperforms event-based control beyond a certain state dimension.
Furthermore, we observe that the performance ratio is lower bounded by the one for the $2$-norm case.

Thirdly, we also simulated the system under $8$-norm triggering.
Note that $2$-, $8$- and $\infty$-norm naturally yield equal cost ratios for $n=1$.
We can already observe a significant performance disadvantage of the $8$- and $\infty$-norm compared to the $2$-norm triggering condition for larger single digit state dimensions.
Moreover, we find that $8$-norm level-triggered control also seems to suffer from the problem that periodic control can outperform it for a sufficiently large state dimension.
Given the simulation results, this state dimension is larger than in the $\infty$-case and seems to lie somewhere around $n\approx 70$ as opposed to $n\approx 40$.

In order to examine the relative performance disadvantage of triggering with other $p$-norms than the $2$-norm more closely, we plot the cost ratios $J_p/J_2$ under the constraint $\E{\tau_p}=\E\Ttwo$ for various $p$ in Fig.~\ref{fig:cost_ratio_wrt_2norm}.
We obtain them by computing the ratio of the corresponding $p$-norm and $2$-norm results from Fig.~\ref{fig:cost_ratio}.
We observe that the ratios $J_p/J_2$ increase rapidly beyond $n=1$ for all simulated $p$-norms.
Moreover, the performance disadvantage of $p$-norm triggering compared to 2-norm triggering grows with increasing $p>2$.

\begin{figure}
	\centering
	\input{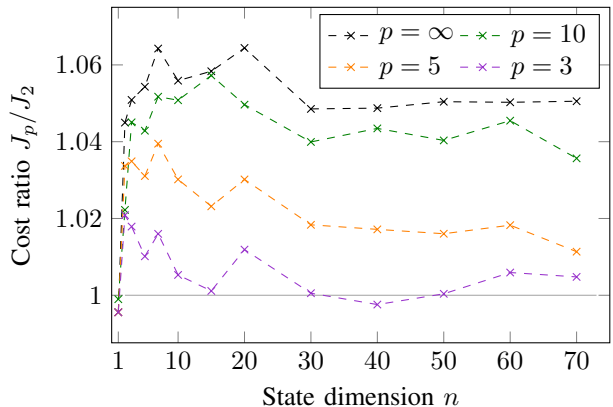}
	\setlength{\abovecaptionskip}{0pt}
	\caption{Cost ratio $J_p/J_2$ for different $p$-norms in the triggering condition under the constraint $\E{\tau_p}=\E\Ttwo$.}
	\label{fig:cost_ratio_wrt_2norm}
\end{figure}

Consequently, the simulation results confirm our theoretical findings.
As demonstrated for the $8$- and $\infty$-norm case, the observed performance degradation for $p$-norms with $p\neq2$ can even render the event-based control scheme inconsistent depending on the state dimension.
Moreover, the simulation results point to the conclusion that the choice of the $p$-norm has a performance implication for any choice of $p\in[1,\infty]$.
In particular, we formulate the conjecture that
\begin{conj}
	The performance of the considered level-triggered sampling schemes degrades with increasing $p>2$.
\end{conj}

\section{Conclusion}\label{sec:conclusion}

In this work, we have examined Conjecture~\ref{conj:p_norm} from \cite[Rem.~1]{Goldenshluger2017}, namely that the selected $p$-norm has no impact on the performance of the level-triggered sampling scheme in an $n$-dimensional single-integrator setup.
We have falsified the conjecture by our theoretical analysis in the $\infty$-norm case.
The provided simulation results confirm this result and, moreover, indicate that the conjecture does not hold for other $p$-norms either.
In addition, based on the performed Monte Carlo simulations, we formulated the new conjecture that the performance of the $p$-norm level-triggering schemes degrades with increasing $p>2$ in the considered setup.

These findings also provide important insights into the source of the phenomenon found in \cite{Meister2022}, namely that event-based control might be outperformed by periodic control in a particular multi-agent system consensus setup.
With the results derived in this paper, we demonstrate that the decentralized triggering condition in \cite{Meister2022}, which can also be written in terms of the $\infty$-norm, is at the root of this performance degradation.
This also reveals that the observed phenomenon is not a direct consequence of considering a cooperative control goal.
Consequently, future research on the phenomenon should focus on examining decentralized triggering conditions and their performance implications.

In future work, we plan to provide theoretical results on the performance of other $p$-norm triggering conditions.
While their behavior has already been examined in simulation in this work, characterizing it theoretically remains an open problem.
In addition, exploring consistency of decentralized triggering conditions is at the core of explaining and potentially resolving the phenomenon found in \cite{Meister2022}.

\appendix

\subsection{Proof of Lemma~\ref{lem:scaling}}
\label{app:proof_scaling}

Let us show the first identity.
Since triggering condition \eqref{eq:ET_cond} satisfies the assumptions of Lemmas~\ref{lem:0T} and \ref{lem:Q}, we have
\begingroup		%
\allowdisplaybreaks
\begin{align*}
	    & Q_p(\Delta_p)/n                                                                                                                     \\
	={} & \E{\int_0^{\tau_p} v_1(s)^2 \diff s}
	\\
	={} & \E{\int_0^{\inf\{ t>0 \;:\; \lVert v_k(\Delta_p^2 t / \Delta_p^2)\rVert{}_p = \Delta_p\}} v_1(\Delta_p^2 s/\Delta_p^2)^2 \diff s}
	\\
	={} & \E{\int_0^{\inf\{ t>0 \;:\; \Delta_p \lVert v_k( t / \Delta_p^2)\rVert{}_p = \Delta_p\}} \Delta_p^2 v_1(s/\Delta_p^2)^2 \diff s}
	\\
	={} & \Delta_p^2 \, \E{\int_0^{\Delta_p^{-2} \inf\{ \Delta_p^2 t' >0 \;:\; \lVert v_k(t')\rVert{}_p = 1\}} v_1(s')^2 \Delta_p^2 \diff s'}
	\\
	={} & \Delta_p^4 \, \E{\int_0^{\inf\{ t' >0 \;:\; \lVert v_k( t')\rVert{}_p = 1\}}  v_1(s')^2 \diff s'}
	\\
	={} & \Delta_p^4 \, Q_p(1)/n.
\end{align*}
\endgroup
The third step leverages the scaling property of Brownian motions.
In the fourth step, we applied linear integral substitution.
The other formula is proved similarly.

\subsection{Proof of Lemma~\ref{lem:ET_inf_var}}
\label{app:proof_ET_inf_var}

For a random variable $Y$ and a real number $M>0$, define its saturated version
\begin{equation*}
	Y^M := \begin{cases}
		Y,                    & |Y|\leq M,        \\
		M \,\mathrm{sign}(Y), & \text{otherwise}.
	\end{cases}
\end{equation*}
It is well known that a Lipschitz transform with Lipschitz constant 1 (or smaller) reduces the variance, hence $\V{Y^M} \leq \V{Y}$.

By the continuous mapping theorem (cf. \cite[Thm.~2.3]{Vaart2000}), $Z_n^M$ weakly converges to $Z^M$ since $Z_n$ weakly converges to $Z$.
Therefore, $\V{Z_n^M}$ converges to $\V{Z^M}$ when $n\to\infty$, as the random variables are bounded.
Thus, it follows that
\begin{equation*}
	\liminf_{n\to\infty} \V{Z_n} \geq \liminf_{n\to\infty} \V{Z_n^M} = \V{Z^M}.
\end{equation*}
With $\V Z < \infty$, letting $M\to\infty$ shows the claim since
\begin{equation*}
	\lim_{M\to\infty} \V{Z^M} = \V Z.
\end{equation*}

\subsection{Proof of Proposition~\ref{prop:sim}}
\label{app:proof_sim}

The key observation is that
\begin{equation*}
	Q(\tau) = \frac{1}{2(n+2)} \, \E{R(\tau)^4},
\end{equation*}
for any stopping time $\tau$ satisfying the assumptions in Lemmas~\ref{lem:0T} and \ref{lem:Q}.
Following along the lines of \cite[Lem.~7]{Meister2023}, this identity can be shown by considering the process $(R(t))_{t\in\lbrack 0,\infty)}$ solving the stochastic differential equation
\begin{equation*}
	\diff R(t) = \frac{n-1}{2R(t)} \diff t + \diff w(t),
\end{equation*}
where $w(t)$ is a standard Wiener process.
Let $Af(x):= \frac{n-1}{2x} f'(x) +\frac{1}{2} f''(x)$ be the infinitesimal generator of the Markov process $(R(t))_{t\in\lbrack 0,\infty)}$.
Then, according to Dynkin's formula, we have for any stopping time $\tau$
\begin{equation*}
	\E{f(R(\tau))} = \E{\int_0^{\tau} A f(R(t)) \diff t}.
\end{equation*}
Choosing $f(x)=x^4$ yields
\begin{equation*}
	\E{R(\tau)^4} = 2(n+2) \E{\int_0^{\tau} R(t)^2 \diff t}.
\end{equation*}
Utilizing Lemma~\ref{lem:Q}, we arrive at
\begin{equation*}
	\E{\int_0^{\tau} R(t)^2 \diff t} = n \, \E{\int_0^{\tau} v_1(t)^2 \diff t} = Q(\tau),
\end{equation*}
which allows us to obtain the stated identity.

Combining this result with Lemma~\ref{lem:0T} and Proposition~\ref{prop:TT_cost} presents us with the desired expression for $J(\tau)/J_\mathrm{per}(\E{\tau})$.

\section*{Acknowledgment}

We thank Frank Aurzada from the Technical University of Darmstadt for the fruitful discussions.

\bibliographystyle{IEEEtran}
\bibliography{IEEEabrv,references}

\end{document}